\PassOptionsToPackage{dvipsnames}{xcolor}
\documentclass[aps,pra,11pt,tightenlines ,preprintnumbers,nofootinbib,superscriptaddress,longbibliography,twocolumn,letterpaper]{revtex4-2}

\usepackage[T1]{fontenc}
\usepackage{tabularx} %
\usepackage{amsmath}  %
\usepackage[dvipsnames]{xcolor}
\usepackage{graphicx} %
\usepackage[normalem]{ulem}
\usepackage[reset, top=0.75in, bottom=0.75in, inner=0.6in, outer=0.6in, letterpaper, columnsep=0.5in]{geometry} %
\usepackage{newtxtext}
\usepackage{physics}
\usepackage{dsfont}
\usepackage{tikz}
\usepackage{amssymb }
\usepackage[final]{hyperref} %
\usepackage{amsmath, amsthm, amssymb}
\usepackage[dvipsnames]{xcolor}
\usepackage{tikz}
\usepackage{graphicx}
\usepackage{subfigure}
\usepackage{afterpage}
\usepackage{capt-of}
\usepackage{dblfloatfix}
\usepackage{mathtools}
\usepackage[autostyle=true]{csquotes}

\usepackage{etoolbox}
\colorlet{Mycolor1}{green!10!orange}
\newtheorem{prop}{Proposition}
\hypersetup{
	colorlinks=true,       %
	linkcolor=magenta,        %
	citecolor=blue,        %
	filecolor=magenta,     %
	urlcolor=blue
}

\newcommand{\comment}[1]{}
\usepackage{paralist} %
\usepackage{ulem}

\usepackage[babel]{microtype} 
\microtypecontext{spacing=nonfrench}
\microtypesetup{
expansion={true,nocompatibility},
protrusion={true,nocompatibility},
activate={true,nocompatibility},
tracking=false,
kerning=true,
spacing={true}
}

\usepackage[all=normal,floats=tight,mathspacing=tight,wordspacing=tight,paragraphs=normal,tracking=tight,charwidths=tight,mathdisplays=normal,sections=normal,margins=normal]{savetrees}

\makeatletter
\renewcommand\@makefnmark{\hbox{\@textsuperscript{\normalfont\color{Green}\@thefnmark}}}
\renewcommand\and{,\penalty-1\ }  %
\makeatother

\theoremstyle{plain}
\newtheorem{theorem}[prop]{Theorem}
\newtheorem{lemma}[prop]{Lemma}
\newtheorem{cor}[prop]{Corollary}
\begin{document}
\newcommand{\maria}[1]{{\color{blue} #1}}
\newcommand{\dani}[1]{{\color{red} #1}}
\newcommand{\yasamin}[1]{{\color{orange} #1}}
\newcommand{\soham}[1]{{\color{green}
#1}}

\title{\vspace*{0.001in} 
Upper Bounding Hilbert Space Dimensions which can Realize all the Quantum Correlations
\vskip 0.1in}
\author{Yasamin Panahi}
\affiliation{Perimeter Institute for Theoretical Physics, Waterloo, Ontario, Canada, N2L 2Y5}
\affiliation{Department of Physics and Astronomy, University of Waterloo, Waterloo, Ontario, Canada, N2L 3G1}
\author{Maria Ciudad Alañón}
\affiliation{Perimeter Institute for Theoretical Physics, Waterloo, Ontario, Canada, N2L 2Y5}
\affiliation{Department of Physics and Astronomy, University of Waterloo, Waterloo, Ontario, Canada, N2L 3G1}
\author{Daniel Centeno}
\affiliation{Perimeter Institute for Theoretical Physics, Waterloo, Ontario, Canada, N2L 2Y5}
\affiliation{Department of Physics and Astronomy, University of Waterloo, Waterloo, Ontario, Canada, N2L 3G1}
\author{Ralph Jason Costales}
\affiliation{Perimeter Institute for Theoretical Physics, Waterloo, Ontario, Canada, N2L 2Y5}
\affiliation{Department of Physics and Astronomy, University of Waterloo, Waterloo, Ontario, Canada, N2L 3G1}
\affiliation{Department of Physics and Astronomy, University College London, London, United Kingdom, WC1E 6BT}
\affiliation{London Centre for Nanotechnology, London, United Kingdom, WC1H 0AH}
\author{Luca Mrini}
\affiliation{Perimeter Institute for Theoretical Physics, Waterloo, Ontario, Canada, N2L 2Y5}
\affiliation{Department of Physics and Astronomy, University of Waterloo, Waterloo, Ontario, Canada, N2L 3G1}
\affiliation{Department of Mathematics, University of Vienna, Strudlhofg. 4 A-1090 Wien, Austria}
\author{Soham Bhattacharyya}
\affiliation{Perimeter Institute for Theoretical Physics, Waterloo, Ontario, Canada, N2L 2Y5}
\affiliation{School of Physical Sciences, Indian Association for the Cultivation of Science, Kolkata-700032, India}
\author{Elie Wolfe\textsuperscript{*}}
\affiliation{Perimeter Institute for Theoretical Physics, Waterloo, Ontario, Canada, N2L 2Y5}
\affiliation{Department of Physics and Astronomy, University of Waterloo, Waterloo, Ontario, Canada, N2L 3G1}

\begin{abstract}
We introduce novel upper bounds on the Hilbert space dimensions required to realize quantum correlations in Bell scenarios. We start by considering bipartite cases wherein one of the two parties has two settings and two outcomes. Regardless of the number of measurements and outcomes of the other party, the Hilbert space dimension of the first party can be limited to two while still achieving all convexly extremal quantum correlations. We then leverage Schmidt decomposition to show that the remaining party can losslessly also be restricted to a qubit Hilbert space. We then extend this idea to multipartite scenarios. We also adapt our results to provide upper bounds of local Hilbert space dimensions to achieve \emph{any} quantum correlation, including convexly non-extremal correlations, by utilizing Carathéodory’s theorem. Finally, we generalize our results to nonstandard Bell scenarios with communication. Taken together, our results fill in several previously unresolved aspects of the problem of determining sufficient Hilbert space dimensionality, expanding the collection of scenarios for which finite-dimensional quantum systems are known to be sufficient to reproduce any quantum correlation.
\end{abstract}
\maketitle

\begingroup
\renewcommand\thefootnote{}\footnotetext{
\textsuperscript{*}{ewolfe@perimeterinstitute.ca}
}
\addtocounter{footnote}{-1}
\endgroup
\section{Introduction}\label{s1}
In 1964, John Stewart Bell's seminal paper \enquote*{On the Einstein-Podolsky-Rosen Paradox} claimed that quantum mechanics cannot be reconciled with local hidden-variable theories~\cite{bell1964einstein}. According to Bell's theorem, there exist inequality constraints - Bell inequalities - which, if violated, no local hidden variable model can adequately describe the experiment and a non-classical explanation is needed to account for the observed statistics. In other words, Bell's theorem says that the set of correlations which are predicted as possible by quantum theory, $\mathcal{Q}$, is strictly larger than the set of correlations predicted as possible by local hidden variable models, $\mathcal{L}$.

Bell scenarios~\cite{loubenets2011upper, Sikora_2019, jebarathinam2019maximal,pal2009quantum} are important from a fundamental physics perspective as they provide a framework where it is possible to characterize the classical versus non-classical distinction. An $N$-partite Bell scenario (Fig.\hyperref[f1]{1}) describes an experiment where a quantum system prepared in a certain initial state is distributed among $N$ space-like separated parties. Then, the $N$ parties perform a local measurement, which depends on a local setting, on their respective shares. The number of different settings and outcomes of the measurements might be different for each player. If the $i^\text{th}$ party has a setting ${j_i \in \{1,...,m_i\}}$ with a corresponding outcome $o_i \in \{1,...,k_{j_i}\}$, we specify this scenario with the notation of  Rosset et al.~\cite{rosset2014classifying}: %
\begin{equation}\label{1}
        [(k_{11}, ..., k_{1m_1}), (k_{21}, ..., k_{2m_2}), ..., (k_{N1}, ..., k_{Nm_N})]. \nonumber
\end{equation}

In these scenarios, when the number of settings and outcomes is finite, the correlations that can be explained using a local hidden variable model form a polytope, known as the local polytope~\cite{Geometry_2018}. Mathematically, every correlation inside the local polytope can be written as, 
\begin{equation}\label{2}
    P(a_1,...,a_N|x_1,...,x_N) = \sum_{\lambda}g(\lambda)\prod_{i=1}^N P_i(a_i|x_i \lambda)
\end{equation}
where $(a_1,\dots,a_N)$ and $(x_1,\dots,x_N)$ are the outputs and measurement settings respectively and $g(\lambda)$ is a distribution over the hidden variable $\lambda$.
The set of quantum correlations, $\mathcal{Q}$, is also a convex set which includes the local polytope. In quantum theory, correlations are given by the Born rule
\begin{equation}\label{3}
    P(a_1,...,a_N|x_1,...,x_N) = \Tr[\rho\hspace{1mm}\bigotimes_{i=1}^N \hat{A_i} (a_i|x_i)],
\end{equation}
where $\rho$ is the density operator for the $N$-partite quantum state and $\hat{A}_i(a_i|x_i)$ is the Positive Operator-Valued Measure (POVM) for the $i^{th}$ party.
If the number of settings and outcomes are restricted to two, it is proven that a local Hilbert space dimension (HSD) of 2, i.e. a qubit, is sufficient for each party to realize all \emph{convexly extremal points} \footnote{A convexly extremal point is defined as a point in the set of quantum correlations that cannot be realized as a convex combination of other points in the set.}of the set of quantum correlations~\cite{masanes2005extremal,masanes2006asymptotic,Pironio_2009,Irfan_2020}. This baseline understanding sets the groundwork for investigating more complex scenarios. The correlation capabilities of qubits, while adequate for the [$\underbrace{(2,2),(2,2),...}_{\text{\textit{N} times}}$] scenario, may no longer suffice when additional degrees of freedom are introduced in the number of settings and outcomes.\\
\begin{figure}[h!]\label{f1}
    \centering
    \includegraphics[width=0.8\linewidth]{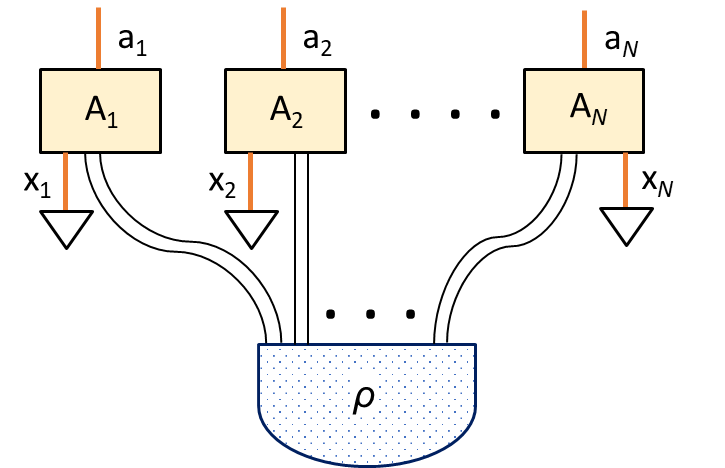}
    \caption{Circuit interpretation for the $N$-partite Bell scenario: A quantum state $\rho$ is distributed to $N$ parties with settings $x_1,x_2,...,x_N$ and outcomes $a_1,a_2,...,a_N$.}
    \label{fig:multipartiteBellcircuit}
\end{figure}

A straightforward question is whether finite HSDs are always sufficient to achieve the convexly extremal correlations in a scenario involving a finite number of settings and outcomes. The answer turns out to be negative, as evidenced by examples in the literature that require infinite HSDs for maximal violation of tight Bell inequalities.\footnote{It has also been shown that the set of quantum correlations for specific graph structures is not closed, indicating that these extremal quantum correlations, which lie outside the closure, can only be achieved using infinite-dimensional systems~\cite{dykema2019non}.} Ref.~\cite{pal2010maximal} provided a protocol using infinite HSDs to achieve the maximal quantum violation of the $I_{3322}$ Bell inequality -- corresponding to the $[(2,2,2),(2,2,2)]$ scenario, first presented in~\cite{Collins2004} -- and conjectured that it cannot be attained using a quantum state of finite dimension. Further evidence can be found in Ref.~\cite{coladangelo2020inherently} which considers a scenario $[(3,3,3,3,3),(3,3,3,3)]$ where two parties have five and four settings respectively, with three outcomes per setting, and identifies a concrete correlation that cannot be achieved by any finite HSDs. 

For a correlation that is \emph{classically} achievable, we can always provide an upper bound on the shared randomness required to achieve it~\cite{nonconvex}. As said above, this is not true in the quantum paradigm because there are examples of correlations which need unbounded HSDs to achieve them. This raises natural questions such as: for which scenarios are finite HSDs sufficient? Where exactly does one draw the line? 

In some scenarios there are \emph{specific} Bell inequalities known to be maximally violable using finite HSDs. Ref.~\cite{pal2009quantum} finds more than two hundred tight Bell inequalities (bipartite, up to five settings, always two outcomes) are maximally violated using local HSDs not exceeding six. However, those results do not demonstrate that finite HSDs are sufficient for other Bell inequalities, even within the same scenario. 

Nevertheless, there are a precious few scenarios for which we know that finite HSDs are sufficient to achieve the convexly extremal points of their quantum set. The most significant case is the aforementioned scenario wherein $N$ players perform two dichotomic measurements; there, qubits are enough~\cite{masanes2005extremal,Pironio_2009,Irfan_2020,masanes2006asymptotic} as we will review in Sec.~\ref{s2.a}. If all convexly extremal quantum correlations in a given scenario can be realized by some finite HSDs, then \emph{all} quantum correlations in the scenario -- extremal and nonextremal both -- can also be realized by some different (higher) finite HSDs, as discussed in Sec.~\ref{s2.d}.

To recap: Some Bell scenarios realize all quantum correlations with known finite HSDs, other Bell scenarios require unbounded HSDs. For everything else, realizability by finite HSDs remains an open question. Here we contribute a small advancement to this open question, expanding the range of scenarios known to have the property of realizing all their quantum correlations with finite HSDs.\footnote{Hereafter, whenever we speak of correlations, we mean quantum correlations.} 

Note that this open question is highly nontrivial. There are many examples of scenarios for which \emph{lower bounds} on HSDs have been identified. Such lower bounds are often expressed as dimension witnesses, which follow the form of Bell inequalities. An instance is~\cite{PhysRevLett.98.010401} which demonstrates that the CGLMP inequalities~\cite{collins2002bell} for $k$ outcomes are not maximally violated by $d$-dimensional qudits for $d<k$ at least in the set $3\leq k \leq 8$. While it turns out that the $k$-outcome CGLMP inequality is maximally violated by $k$-dimensional Hilbert spaces for each party, it is unknown if \emph{every} Bell inequality for such $k$-outcome two-setting bipartite scenarios can also be maximally violated by $k$-dimensional Hilbert spaces. For the tripartite case, however, it turns out that $k$-dimensional Hilbert spaces are generally \emph{insufficient}, even when restricting to two settings per party. Ref.~\cite{PhysRevA.85.052113} identified a special facet-defining tripartite Bell inequality wherein each party performs measurements with two settings and three outcomes; they numerically proved that local HSDs of $\vec{d}{=}\{6,6,6\}$ provided significantly greater violation of the inequality relative to quantum strategies limited to $\vec{d}{=}\{3,3,3\}$. In other words: patterns have exceptions, and rigorous upper bounds of HSDs such as those in this work cannot be taken for granted.

The structure of the paper is as follows. In Sec.\ref{s2}, we review existing results that establish upper bounds on HSDs in standard Bell scenarios, and then present our main result, which extends these bounds to more Bell scenarios (i.e. scenarios with different number of parties, settings and outcomes). In Sec.\ref{s3}, we highlight the connection between standard Bell scenarios and those involving communication, which enables us to apply our bounds on HSDs to this broader class of scenarios. Finally, in Sec.~\ref{s4}, we conclude with a brief summary of our findings.

\section{Bell Scenarios}\label{s2}
\subsection{Existing Results}\label{s2.a}
In the bipartite scenario $[(2,2),(2,2)]$, it has been proven that qubits are sufficient to produce all convexly extremal correlations~\cite{masanes2005extremal,masanes2006asymptotic,Pironio_2009,Irfan_2020}. In fact, this result generalizes straightforwardly so that qubits are sufficient to achieve all convexly extremal correlations in the case with $N$ parties $[\underbrace{(2,2),(2,2),...}_{\text{\textit{N} times}}]$. This generalization is stated formally in a proposition due to Masanes, which we quote here.
\begin{prop} [\textit{Masanes~\cite{masanes2005extremal}, Theorem 4}]\label{pr1}
In the Bell scenario with N parties each having two dichotomic observables, all the convexly extremal quantum correlations are achievable by measuring $N$-qubit pure states with projective observables.
\end{prop}
However, an alternative formulation can also be provided. It was presented as an intermediate result in the proof of the original theorem~\cite{masanes2005extremal,masanes2006asymptotic}. Specifically, Masanes' result explicitly demonstrates that in a Bell scenario, a single party with two settings and two outcomes can always be described by a qubit, regardless of the number of other parties or their respective number of settings and outcomes.
\begin{prop}\label{pr2}
If there is a convexly extremal quantum correlation in a $N$-partite Bell scenario where one of the parties (say Alice) has two settings-two outcomes, then that correlation can be achieved with Alice having a local Hilbert space of dimension 2.
\end{prop}

It follows directly from this proposition that, whenever the local Hilbert space dimensions (HSDs) of the other parties are finite, the correlation can be realized with Alice having a local Hilbert space dimension of 2, while the remaining parties retain their respective finite HSDs.
\subsection{Bipartite Scenario}\label{s2.b}
Suppose we have a bipartite Bell scenario (Alice and Bob) where Bob has two dichotomic observables and Alice has an arbitrary number of settings and outcomes. Proposition~\ref{pr2} tells us that Bob's Hilbert space can be taken to be 2-dimensional to achieve the convexly extremal correlations. Notably, Alice's local Hilbert space dimension can also be restricted without loss of generality by virtue of Schmidt decomposition~\cite{Sikora_2016}. The Schmidt decomposition can be used to express a bipartite pure quantum state as one in a Hilbert space whose dimension is restricted to the minimum of the dimensions of the two parts of the state.
\begin{figure}[h!]
    \centering
    \includegraphics[width=0.4\linewidth]{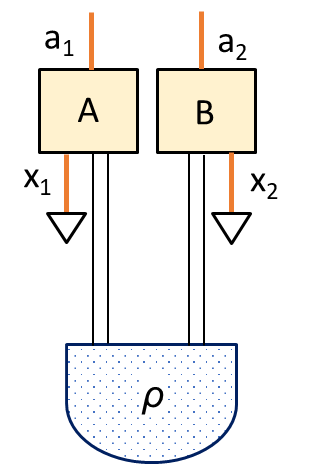}
    \caption{Circuit interpretation for bipartite Bell scenario.}
    \label{fig:bipartiteBellcircuit}
\end{figure}
\begin{lemma}[Schmidt Decomposition]\label{lem:schmidt}
Let $\ket{\psi}$ be a pure state in the tensor product Hilbert space $\mathcal{H}_A \otimes \mathcal{H}_B$, with local Hilbert spaces $\mathcal{H}_A$ {and} $\mathcal{H}_B${ of dimension} $d_A$ {and} $d_B$ respectively. Let $d'\coloneqq \min(d_A, d_B)$.  Then, there exist orthonormal sets  $\{ \ket{u_1},...,\ket{u_{d'}} \} \subset \mathcal{H}_A$ {and} $\{ \ket{v_1},...,\ket{v_{d'}} \} \subset \mathcal{H}_B$ {such that} $\ket{\psi} = \sum_{i=1}^{d'} \lambda_i \ket{u_i} \otimes \ket{v_i}$ {where the scalars} $\lambda_i$ {are real, strictly positive, and normalized such that} $\sum_i \lambda_i^2 = 1$.
\end{lemma}
Combining the use of the Schimdt decomposition with Masanes' proposition, we can state that all the convexly extremal correlations in such a scenario can be realized using just qubits for each party. 
The following general result, previously noted in Ref.~\cite{Sikora_2016}, is worth highlighting:
\begin{lemma}
\label{cor2.1}
    Consider a quantum correlation in the bipartite Bell scenario, where the HSD for one of the parties, say Alice, is $d_A$. Then, the HSD of the other party, say Bob, is less than or equal to the one required by Alice, $d_B\leq d_A$.
\end{lemma}
Notice that the lemma implies that if we know that a quantum correlation can be produced with HSDs $d_A$ and $d_B$ for Alice and Bob, respectively, then it can also be produced with HSD $d=\mbox{min}\{d_A,d_B\}$ for both parties.
\begin{proof}
Suppose we have a quantum distribution, $P(a_1,a_2|x_1,x_2)$ as specified in Eq. (\hyperref[3]{3}) with ${N=2}$. Without loss of generality we will consider pure states. Note that we can always purify a state by increasing the HSD and if the HSD of the mixed state is finite, the purification of it will also have finite HSD.
Consider the Schmidt decomposition of a bipartite pure state as per Lemma~\ref{lem:schmidt}. Then, we can define a \emph{Schmidt projector} as follows:
\begin{equation}\label{4}
    \Pi \coloneqq \sum_{i=1}^{d'} \ket{u_i}\bra{u_i} \otimes \ket{v_i}\bra{v_i}.
\end{equation}
Note that a bipartite pure state ${\rho} = \ket{\psi}\bra{\psi}$, with $\ket{\psi} \in \mathcal{H}_A \otimes \mathcal{H}_B$, remains invariant under the action of this projection operator. However, it can act non-trivially on the measurement restricting the dimension of the vector space on which they operate, while ensuring that the overall correlation remains invariant.
That is, 
if a conditional probability $P(a_1,a_{2}|x_1,x_{2})$ with outcomes $a_1,a_2$ and settings $x_1,x_2$ can be realized by some application of local measurement (i.e., $\hat{M}_{a_1|x_1}$ and $\hat{M}_{a_2|x_2}$) on a bipartite pure state $\rho $ then, 
\begin{align}\begin{split}\label{5}
    &P(a_1,a_{2}|x_1,x_{2}) \\
    &= \Tr[\rho\vdot( M_{a_1|x_1} \otimes M_{a_{2}|x_{2}})] \\
    &= \Tr[\Pi \rho \Pi\vdot( M_{a_1|x_1} \otimes M_{a_{2}|x_{2}})]  \\
    &=\Tr[\rho \vdot\Pi( M_{a_1|x_1} \otimes M_{a_{2}|x_{2}})\Pi]  \\
    &=\Tr[\rho\vdot( \tilde{M}_{a_1|x_1} \otimes \tilde{M}_{a_{2}|x_{2}})] \\
\end{split}
\end{align}
where
\begin{align*}
    \tilde{M}_{a_1|x_1} \coloneqq \sum_{i=1}^{d'} \sum_{j=1}^{d'} \ket{u_i} \bra{u_i} \hat{M}_{a_1|x_1} \ket{u_j} \bra{u_j}, \nonumber \\
    \tilde{M}_{a_2|x_2} \coloneqq \sum_{i=1}^{d'} \sum_{j=1}^{d'} \ket{v_i} \bra{v_i} \hat{M}_{a_2|x_2} \ket{v_j} \bra{v_j}.
\end{align*}
are the new local measurement operators that act only on the $d'$ dimensional subspace of $\mathcal{H}_{A}$ and $\mathcal{H}_{B}$, respectively.
The proof exploits the invariance of the state under the action of the $\hat{\Pi}$ and the cyclicity of trace in the second and third steps respectively.\\ 
\end{proof}
It is interesting to apply Lemma~\ref{cor2.1} to the bipartite Bell scenario where one party, say Alice, has 2 settings and 2 outcomes. It is particularly useful for this scenario because we know from Proposition~\ref{pr2} that, without loss of generality, Alice's Hilbert space dimension can always be set to two, $d_A=2$, in order to achieve all the quantum convexly extremal correlations. Then by Lemma~\ref{cor2.1}, one concludes that a two-dimensional Hilbert space for Bob is sufficient to achieve those correlations for any arbitrary number of settings and outcomes on his side.

\subsection{Multipartite Scenarios}\label{s2.c}
We extend our previous result, Lemma~\ref{cor2.1}, to $N$-partite Bell scenarios. Given a certain quantum correlation for which the required local HSDs for the first ${N{-}1}$ parties are known, then we derive an upper bound for the HSD of the last party. This is formalized in the following lemma.
\begin{lemma}
\label{cor2.2}
    Consider a quantum correlation in a $N$-partite Bell scenario where the HSDs of the first ${N{-}1}$ parties are $d_1, d_2, ..., d_{N{-}1}$, respectively. Then, the necessary HSD of the last party is upper bounded as $d_N \leq d_1\times d_2 \times \cdots \times d_{N{-}1}$.
\end{lemma}
\begin{proof}
Suppose we have the distribution as specified in $\text{Eq. (\hyperref[3]{3})}$. Without loss of generality, we can work with pure states using the same reasoning as in the bipartite case. Consider a pure state $\ket{\psi}\in \mathcal{H}_{\bar{A}}\otimes\mathcal{H}_{A_N}$, where $\mathcal{H}_{A_N}$ is the Hilbert space for the last party with dimension $d_N$ and $\mathcal{H}_{\bar{A}} = \bigotimes _{i=1}^{N{-}1} \mathcal{H}_{A_i}$ is the joint Hilbert space for the first ${N{-}1}$ parties (we denote the $i^{th}$ party's local Hilbert space by $\mathcal{H}_{A_i}$ for $i \in \{1,...,N{-}1\}$)  with dimension  $d=d_1 \times d_2...\times d_{N{-}1}$. This can be seen as a bipartite state where the first ${N{-}1}$ parties are grouped into an effective first party, which allows us to re-use the idea of Lemma~\ref{cor2.1}. Schmidt decomposition tells us that we can always find a $d$-dimensional subspace of $\mathcal{H}_{A_N}$ which along with $\mathcal{H}_{\bar{A}}$ suffices to represent the quantum state and the further steps (as in Lemma~\ref{cor2.1}) lead to the result that a local Hilbert space of dimension $d'_N = d = d_1 \times d_2...\times d_{N{-}1}$ is sufficient for the last party in order to achieve the given correlation.  
\end{proof}
One particularly interesting example is a $N$-partite scenario where the first ${N{-}1}$ parties have two settings and two outcomes. Applying Proposition~\ref{pr2}, we know that qubits are sufficient for the first ${N{-}1}$ parties in order to achieve every convexly extremal quantum correlation. Furthermore, Lemma~\ref{cor2.2} demonstrates that the local HSD of the last party is less than or equal to $d=2^{N{-}1}$. Rephrasing this fact as our first practical consequence:
\begin{theorem}
Suppose $f$ is a convex function of the correlations in an $N$-partite Bell scenario where the first ${N{-}1}$ parties have two settings and two outcomes. Then, the maximum value of $f$ over \emph{all} quantum correlations is equal to the maximum value of $f$ optimized over the set of quantum correlations generated by measurements on a \emph{pure state} with the first ${N{-}1}$ local Hilbert space dimensions equal to $2$ and the final local Hilbert space dimension equal to $2^{N{-}1}$.
\label{th1}
\end{theorem}
\subsection{Realizing convexly non-extremal correlations}\label{s2.d}

Let $\mathcal{Q}_{\vec{d}}$ 
 represent the set of all correlations realizable in a given Bell scenario by measuring $N{-}\text{partite}$ quantum states using local Hilbert space dimensions of at most $\vec{d}$. That is, $\vec{d}$ denotes a list of local Hilbert space dimensions. In this notation, the set of \emph{all} quantum correlations for that Bell scenario is given by $\mathcal{Q}_{\vec{\infty}}$, meaning the set of correlations realizable without any finite restriction to any of the local Hilbert space dimensions. 
Is there some list of finite Hilbert space dimensions $\vec{d^{\star}}$ such that $\mathcal{Q}_{\vec{\infty}} = \mathcal{Q}_{\vec{d^{\star}}}$? Thus far our upper bounds on Hilbert space dimensions only apply to \emph{convexly extremal} quantum correlations. But what if we want to bound the dimensions required to explain \emph{any} quantum correlation, convexly extremal or not?
Our prior results allow to say that -- for certain Bell scenarios -- there exists a finite dimension $d$ such that
\begin{subequations}
\begin{align}
\mathcal{Q}_{\vec{d}}\subseteq\mathcal{Q}_{\vec{\infty}} &\quad\text{yet}\quad\textsf{ExtremalPoints}(\mathcal{Q}_{\vec{\infty}}) \subset \mathcal{Q}_{\vec{d}}\,.
\shortintertext{Equivalently, }\label{eq:convexhull}
 &\mathcal{Q}_{\vec{\infty}} = \textsf{ConvexHull}(\mathcal{Q}_{\vec{d}})\,.
\end{align}
\end{subequations}
Our desideratum here, however, is a \emph{different} dimension bound in order to say $\mathcal{Q}_{\vec{\infty}} = \mathcal{Q}_{\vec{d^{\star}}}$ for the same Bell scenarios.
Nevertheless, we can use physical convexification procedures to obtain upper bounds on such a $\vec{d^{\star}}$ from upper bounds on $\vec{d}$.
We will show that Eq.~\eqref{eq:convexhull} implies that
\begin{align}\begin{split}
&\mathcal{Q}_{\vec{\infty}} = \mathcal{Q}_{\vec{d^{\star}}}\\
\text{where } \; &\vec{d^{\star}} = \vec{d}\times\textsf{CathNum}(\mathcal{Q}_{\vec{d}})
\end{split}\end{align}
where we have introduced the notation $\textsf{CathNum}(\mathcal{S})$ to indicate the Carathéodory number of $\mathcal{S}$ for any nonconvex set $\mathcal{S}$. Given any nonconvex set $\mathcal{S}$ the Carathéodory number of $\mathcal{S}$ is the maximum number of distinct points within $S$ which must be convexly combined to express any point in $\textsf{ConvexHull}(\mathcal{S})$.
Here it is worth clarifying how physical convexification achieves the mathematical convex hull of $\mathcal{Q}_{\vec{d}}$. The physical protocol for convexification corresponds to the preparation of classical mixtures of convexly extremal quantum correlations (which are realizable using a quantum state and measurements with local HSD specified by $\vec{d}$). Such mixtures can be implemented by introducing an ancillary quantum system for each party at the source. These ancillae encode classical randomness and determine which state and measurement settings are used in a given realization. By measuring the ancilla prior to performing the local measurements, one selects among a finite set of correlations in $\mathcal{Q}_{\vec{d}}$ according to a classical probability distribution. As a result, any point in $\textsf{ConvexHull}(\mathcal{Q}_{\vec{d}})$  can be physically realized, and the mathematical operation of taking the convex hull is faithfully reproduced.
Carathéodory's eponymous theorem states that for any nonconvex set $\mathcal{S}$ it holds that $\textsf{CathNum}(\mathcal{S})\leq \textsf{AffineDimension}(\mathcal{S})+1$. Fenchel's theorem is a refinement of Carathéodory's theorem which states that if the set $\mathcal{S}$ is pathwise-connected, then $\textsf{CathNum}(\mathcal{S})\leq \textsf{AffineDimension}(\mathcal{S})$. We know that $\mathcal{Q}_{\vec{d}}$ is pathwise connected, and hence we have
\begin{align}
\mathcal{Q}_{\vec{\infty}} = \mathcal{Q}_{\vec{d}\times\textsf{AffineDimension}(\mathcal{Q}_{\vec{d}})}\,.
\end{align}
Happily, the affine dimension of $\mathcal{Q}_{\vec{d}}$ is readily known for any Bell scenario, as $\textsf{AffineDimension}(\mathcal{Q}_{\vec{d}})=\textsf{AffineDimension}(\mathcal{Q}_{\vec{\infty}})$.
If party $i$ has $m_i$ different measurement settings, and the measurement setting $j$ for party $i$ (denoted $j_i$) has $k$ possible outcomes, then per Ref.~\citep[Theorem 1]{Lifting} we have
\begin{equation}\label{8}
    \textsf{AffineDimension}(\mathcal{Q}_{\vec{\infty}}) = \prod_{i=1}^{N} {\left(\sum_{j=1}^{m_i} {\left(k_{j_i}-1\right)}+1\right)}-1
\end{equation}
Putting this all together, we have:
\begin{theorem}
Suppose $p$ is a correlation for some $N$-partite Bell scenario where the first ${N{-}1}$ parties have two settings and two outcomes. Then, $p$ admits \emph{some} quantum realization (i.e., $p\in \mathcal{Q}_{\vec{\infty}}$) if and only if $p$ admits a quantum realization where the first ${N{-}1}$ parties are measuring on local Hilbert space dimensions equal to $2\times\textsf{AffineDimension}(\mathcal{Q}_{\vec{\infty}})$ and the last party is measuring on local Hilbert space dimension equal to $2^{N{-}1}\times\textsf{AffineDimension}(\mathcal{Q}_{\vec{\infty}})$.
\label{th2}
\end{theorem}
\section{Bell Scenarios with Communication (Bell+)}\label{s3}
We are now focusing on processes where a quantum state is distributed among multiple parties but the mechanisms for how each party obtains its outcome are no longer limited to just some settings which are private to that party or even settings which are exogenous. Rather, we allow for richer variety, such as a common setting shared by multiple parties and allowing the settings for one party to depend on the outputs of some other parties. Such processes are known as Bell scenarios with communication, hereafter denoted as \emph{Bell+} scenarios for brevity\cite{Chaves}.
In Bell and Bell+ scenarios, each party utilizes a measurement device which takes into account some classical values (settings) when determining how to perform their measurement on their portion of the quantum state. Unique to Bell+ scenarios, however, is the possibility that these setting values for a party $A$ may not be freely determined solely by $A$, but rather may be restricted to match settings or outputs of the other parties. In this sense, Bell+ scenarios can be considered as restrictions on Bell scenarios. That is, although in a Bell+ scenario, the parties are not freely toggling the settings of their local measurement devices, we can nevertheless imagine a maximal interruption experiment where they can. The scenario constructed in that experiment will be a standard Bell scenario without communication. The probabilities of Bell$+$ will be related to the probabilities of the maximal interruption scenario through a consistency constraint\cite{wolfe2021quantum, van2019quantum}. Let us explain this idea with a simple example. 
\begin{figure}[h!]\label{f3}
    \centering
    \includegraphics[width=0.5\linewidth]{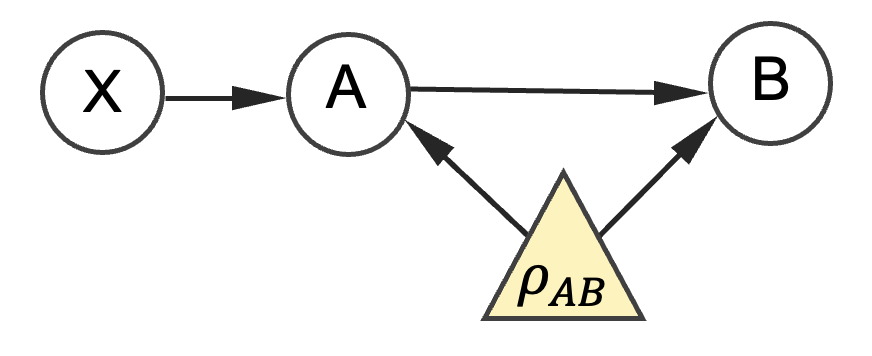}
    \caption{Instrumental scenario}
    \label{fig:instrumental}
\end{figure}
\begin{figure}[h!]\label{f4}
    \centering
    \includegraphics[width=0.5\linewidth]{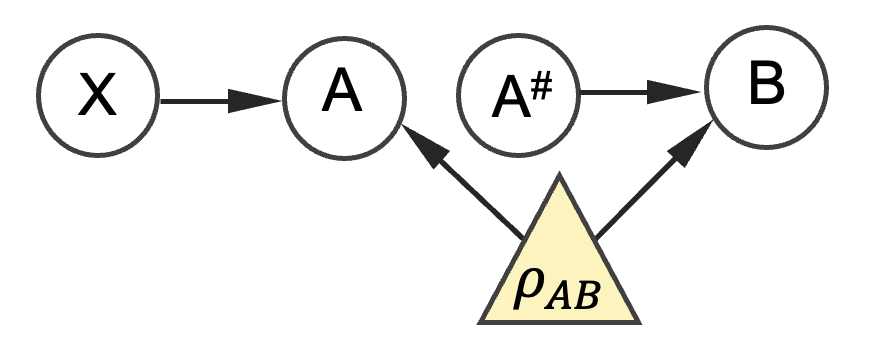}
    \caption{Maximal interruption for the instrumental scenario}
    \label{fig:interruption}
\end{figure}
Consider the Bell+ scenario shown in Fig.~\hyperref[f3]{3}, commonly referred to as the instrumental scenario, which features two parties: Alice ($A$) and Bob ($B$). Unlike in the standard Bell scenario, where measurement settings are chosen independently and privately by each party, the instrumental scenario allows Bob's measurement settings to depend explicitly on Alice's outcomes. To analyze this Bell+ scenario more clearly, we consider a modified version in which the dependency between parties is deliberately interrupted. This leads to the maximal interrupted scenario (Fig.~\hyperref[f4]{4})\cite{wolfe2021quantum}, where the dependency between Alice’s outcome $A$ and Bob’s setting is replaced by an additional exogenous variable $A^{\#}$ that acts only as Bob's setting. The correlations that can be produced in the instrumental scenario can be recovered from the maximal interrupted scenario by post-selecting on $A^{\#} = A$. Geometrically, this corresponds to projecting the higher-dimensional space of the standard Bell scenario onto a strictly smaller subspace that corresponds to the instrumental scenario. Hence, the set of correlations that can be produced in the instrumental scenario can be understood as a projection of the set that can be produced in the standard Bell scenario. Note that this procedure is applicable to any Bell+ scenario.
This geometrical point of view leads to the following general result (as noted in ~\cite{van2019quantum} for the case of the instrumental scenario): \\
\begin{prop}
Any correlation achievable in a Bell+ scenario can be realized using a quantum system whose Hilbert space dimension does not exceed the dimension required for the corresponding maximal interrupted scenario.
\end{prop}

It follows from the projection structure that certain geometric properties of the quantum set may not be retained in the Bell+ scenario. In particular, convex extremality is not preserved under projection. A correlation can only be convexly extremal in the Bell+ scenario if it is the projection of some convexly extremal correlation in the corresponding maximal interruption scenario. However, the converse does not hold, since geometric projection does not preserve convex extremality. This argument allows us to draw the following corollaries based on Thm.\ref{th1} and Thm.\ref{th2}. 
\begin{cor}
    Suppose $f$ is a convex function of the correlations in an $N$-partite Bell+ scenario where the first ${N{-}1}$ parties have two settings and two outcomes. Then, the maximum value of $f$ over \emph{all} quantum correlations is equal to the maximum value of $f$ optimized over the set of quantum correlations generated by measurements on a \emph{pure state} with the first ${N{-}1}$ local Hilbert space dimensions equal to 2 and the final local Hilbert space dimension equal to $2^{N{-}1}$.
\end{cor}
\begin{cor}
Suppose $p$ is a correlation for some $N$-partite Bell+ scenario where the first ${N{-}1}$ parties have two settings and two outcomes. Then, $p$ admits \emph{some} quantum realization (i.e., $p\in \mathcal{Q}_{\vec{\infty}}$) if and only if $p$ admits a quantum realization where the first ${N{-}1}$ parties are measuring on local Hilbert space dimensions equal to $2\times\textsf{AffineDimension}(\mathcal{Q}_{\vec{\infty}})$ and the last party is measuring on local Hilbert space dimension equal to $2^{N{-}1}\times\textsf{AffineDimension}(\mathcal{Q}_{\vec{\infty}})$.
\end{cor}
Note that in these corollaries we are just translating the results from standard Bell scenarios to Bell+ scenarios. In addition, notice that we are giving upper bounds to the local HSD to achieve every point in the quantum set but some of the points can be achieved with smaller HSD, as, for example, all the convexly extremal points in the interrupted version that are non-extremal after the projection.
\section{Conclusion}\label{s4}
The question of determining what quantum resources, in terms of Hilbert space dimensions (HSDs), are required to reproduce given quantum correlations is pertinent to diverse applications in theoretical and applied quantum information. Here, we derive upper bounds on the local HSD required to reproduce convexly extremal quantum correlations in certain multipartite Bell scenarios, concretely those in which ${N{-}1}$ parties have two dichotomic measurements. Additionally, by leveraging the geometry of the quantum correlations set, we establish upper bounds for non-extremal correlations as well. These results are valuable for two distinct applications: the first one facilitates convex optimization over the quantum set of correlations and the second one provides a practical approach to addressing the membership problem, i.e., determining whether a given correlation belongs to the quantum set.

The vast majority of previous studies focused on the necessary local HSD for different versions of the standard Bell scenario, i.e., different cardinalities for settings and outputs. In this work, we point out that all such results in the standard Bell scenario can be translated into similar results in nonstandard Bell scenarios with communication. This generalization is based on the observation that the \emph{maximal interruption} of a nonstandard Bell scenario is itself a standard Bell scenario. As a straightforward consequence, we are able to establish upper bounds on the sufficient local HSD in far more general causal scenarios than those previously known.

To summarize, this article contributes to the ongoing effort to bound the Hilbert space dimensions required to reproduce quantum correlations. This problem holds particular importance for numerical studies, especially optimization tasks such as computing the maximum violation of Bell inequalities, which become significantly more tractable when finite-dimensional realizations are assumed. It should also be emphasized that the distinction between finite and infinite Hilbert space dimensions is not merely technical but reflects a qualitative difference with deep foundational implications. Despite some progress, the problem of finite-dimensional realization has only been addressed in fragments throughout the literature, leaving a substantial gap in our understanding and highlighting the need for future research.

\section{Acknowledgements}\label{s5}
Y.P. would like to acknowledge the support provided by Chong Wang and Timothy Hsieh, through their NSERC Discovery Grants, during the completion of this manuscript. Y.P., R.J.C., and L.M. also thank the Perimeter Scholars International (PSI) Master’s Program for facilitating this work. S.B. thanks the PSI Start Program. Research at Perimeter Institute is supported in part by the Government of Canada through the Department of Innovation, Science and Economic Development and by the Province of Ontario through the Ministry of Colleges and Universities.

L.M. would also like to include the following acknowledgement: This research was funded in whole or in part by the Austrian Science Fund (FWF) [10.55776/EFP6]. For open access purposes, the author has applied a CC BY public copyright license to any author-accepted manuscript version arising from this submission.

R.J.C. would also like to include the following acknowledgement: This work was supported by the Engineering  and Physical Sciences Research Council EP/Y035046/1.

\nocite{*}
\vspace{-0.01in}
\setlength{\bibsep}{.15\baselineskip plus .05\baselineskip minus .05\baselineskip}
\bibliographystyle{apsrev4-2-wolfe}
\nocite{apsrev42Control}
\bibliography{references}

\begin{thebibliography}{41}%
\makeatletter
\providecommand \@ifxundefined [1]{%
 \@ifx{#1\undefined}
}%
\providecommand \@ifnum [1]{%
 \ifnum #1\expandafter \@firstoftwo
 \else \expandafter \@secondoftwo
 \fi
}%
\providecommand \@ifx [1]{%
 \ifx #1\expandafter \@firstoftwo
 \else \expandafter \@secondoftwo
 \fi
}%
\providecommand \natexlab [1]{#1}%
\providecommand \enquote  [1]{``#1''}%
\providecommand \bibnamefont  [1]{#1}%
\providecommand \bibfnamefont [1]{#1}%
\providecommand \citenamefont [1]{#1}%
\providecommand \href@noop [0]{\@secondoftwo}%
\providecommand \href [0]{\begingroup \@sanitize@url \@href}%
\providecommand \@href[1]{\@@startlink{#1}\@@href}%
\providecommand \@@href[1]{\endgroup#1\@@endlink}%
\providecommand \@sanitize@url [0]{\catcode `\\12\catcode `\$12\catcode `\&12\catcode `\#12\catcode `\^12\catcode `\_12\catcode `\%12\relax}%
\providecommand \@@startlink[1]{}%
\providecommand \@@endlink[0]{}%
\providecommand \url  [0]{\begingroup\@sanitize@url \@url }%
\providecommand \@url [1]{\endgroup\@href {#1}{\urlprefix }}%
\providecommand \urlprefix  [0]{URL }%
\providecommand \Eprint [0]{\href }%
\providecommand \doibase [0]{https://doi.org/}%
\providecommand \selectlanguage [0]{\@gobble}%
\providecommand \bibinfo  [0]{\@secondoftwo}%
\providecommand \bibfield  [0]{\@secondoftwo}%
\providecommand \translation [1]{[#1]}%
\providecommand \BibitemOpen [0]{}%
\providecommand \bibitemStop [0]{}%
\providecommand \bibitemNoStop [0]{.\EOS\space}%
\providecommand \EOS [0]{\spacefactor3000\relax}%
\providecommand \BibitemShut  [1]{\csname bibitem#1\endcsname}%
\let\auto@bib@innerbib\@empty
\bibitem [{\citenamefont {Bell}(1964)}]{bell1964einstein}%
  \BibitemOpen
  \bibfield  {author} {\bibinfo {author} {\bibfnamefont {J.~S.}\ \bibnamefont {Bell}},\ }\bibfield  {title} {\enquote {\bibinfo {title} {{On the Einstein Podolsky Rosen paradox}},}\ }\href {https://doi.org/10.1103/PhysicsPhysiqueFizika.1.195} {\bibfield  {journal} {\bibinfo  {journal} {Physics Physique Fizika}\ }\textbf {\bibinfo {volume} {1}},\ \bibinfo {pages} {195--200} (\bibinfo {year} {1964})}\BibitemShut {NoStop}%
\bibitem [{\citenamefont {Loubenets}(2011)}]{loubenets2011upper}%
  \BibitemOpen
  \bibfield  {author} {\bibinfo {author} {\bibfnamefont {Elena~R}\ \bibnamefont {Loubenets}},\ }\bibfield  {title} {\enquote {\bibinfo {title} {{Upper bounds on violation of Bell-type inequalities by a multipartite quantum state}},}\ }\href {https://arxiv.org/abs/1108.0263} {\bibfield  {journal} {\bibinfo  {journal} {arXiv preprint arXiv:1108.0263}\ } (\bibinfo {year} {2011})}\BibitemShut {NoStop}%
\bibitem [{\citenamefont {{Wei}}\ and\ \citenamefont {{Sikora}}(2019)}]{Sikora_2019}%
  \BibitemOpen
  \bibfield  {author} {\bibinfo {author} {\bibfnamefont {Zhaohui}\ \bibnamefont {{Wei}}}\ and\ \bibinfo {author} {\bibfnamefont {Jamie}\ \bibnamefont {{Sikora}}},\ }\bibfield  {title} {\enquote {\bibinfo {title} {{Device-independent dimension test in a multiparty Bell experiment}},}\ }\href {https://doi.org/10.1088/1367-2630/ab1514} {\bibfield  {journal} {\bibinfo  {journal} {New Journal of Physics}\ }\textbf {\bibinfo {volume} {21}},\ \bibinfo {eid} {043021} (\bibinfo {year} {2019})}\BibitemShut {NoStop}%
\bibitem [{\citenamefont {Jebarathinam}\ \emph {et~al.}(2019)\citenamefont {Jebarathinam}, \citenamefont {Hung}, \citenamefont {Chen},\ and\ \citenamefont {Liang}}]{jebarathinam2019maximal}%
  \BibitemOpen
  \bibfield  {author} {\bibinfo {author} {\bibfnamefont {Chellasamy}\ \bibnamefont {Jebarathinam}}, \bibinfo {author} {\bibfnamefont {Jui-Chen}\ \bibnamefont {Hung}}, \bibinfo {author} {\bibfnamefont {Shin-Liang}\ \bibnamefont {Chen}},\ and\ \bibinfo {author} {\bibfnamefont {Yeong-Cherng}\ \bibnamefont {Liang}},\ }\bibfield  {title} {\enquote {\bibinfo {title} {{Maximal violation of a broad class of Bell inequalities and its implication on self-testing}},}\ }\href {https://journals.aps.org/prresearch/abstract/10.1103/PhysRevResearch.1.033073} {\bibfield  {journal} {\bibinfo  {journal} {Physical Review Research}\ }\textbf {\bibinfo {volume} {1}},\ \bibinfo {pages} {033073} (\bibinfo {year} {2019})}\BibitemShut {NoStop}%
\bibitem [{\citenamefont {P{\'a}l}\ and\ \citenamefont {V{\'e}rtesi}(2009)}]{pal2009quantum}%
  \BibitemOpen
  \bibfield  {author} {\bibinfo {author} {\bibfnamefont {K{\'a}roly~F}\ \bibnamefont {P{\'a}l}}\ and\ \bibinfo {author} {\bibfnamefont {Tam{\'a}s}\ \bibnamefont {V{\'e}rtesi}},\ }\bibfield  {title} {\enquote {\bibinfo {title} {{Quantum bounds on Bell inequalities}},}\ }\href {https://journals.aps.org/pra/abstract/10.1103/PhysRevA.79.022120} {\bibfield  {journal} {\bibinfo  {journal} {Physical Review A}\ }\textbf {\bibinfo {volume} {79}},\ \bibinfo {pages} {022120} (\bibinfo {year} {2009})}\BibitemShut {NoStop}%
\bibitem [{\citenamefont {Rosset}\ \emph {et~al.}(2014)\citenamefont {Rosset}, \citenamefont {Bancal},\ and\ \citenamefont {Gisin}}]{rosset2014classifying}%
  \BibitemOpen
  \bibfield  {author} {\bibinfo {author} {\bibfnamefont {Denis}\ \bibnamefont {Rosset}}, \bibinfo {author} {\bibfnamefont {Jean-Daniel}\ \bibnamefont {Bancal}},\ and\ \bibinfo {author} {\bibfnamefont {Nicolas}\ \bibnamefont {Gisin}},\ }\bibfield  {title} {\enquote {\bibinfo {title} {{Classifying 50 years of Bell inequalities}},}\ }\href {https://doi.org/10.1088/1751-8113/47/42/424022} {\bibfield  {journal} {\bibinfo  {journal} {Journal of Physics A: Mathematical and Theoretical}\ }\textbf {\bibinfo {volume} {47}},\ \bibinfo {pages} {424022} (\bibinfo {year} {2014})}\BibitemShut {NoStop}%
\bibitem [{\citenamefont {Goh}\ \emph {et~al.}(2018)\citenamefont {Goh}, \citenamefont {Kaniewski}, \citenamefont {Wolfe}, \citenamefont {V\'ertesi}, \citenamefont {Wu}, \citenamefont {Cai}, \citenamefont {Liang},\ and\ \citenamefont {Scarani}}]{Geometry_2018}%
  \BibitemOpen
  \bibfield  {author} {\bibinfo {author} {\bibfnamefont {Koon~Tong}\ \bibnamefont {Goh}}, \bibinfo {author} {\bibfnamefont {Jędrzej}\ \bibnamefont {Kaniewski}}, \bibinfo {author} {\bibfnamefont {Elie}\ \bibnamefont {Wolfe}}, \bibinfo {author} {\bibfnamefont {Tam\'as}\ \bibnamefont {V\'ertesi}}, \bibinfo {author} {\bibfnamefont {Xingyao}\ \bibnamefont {Wu}}, \bibinfo {author} {\bibfnamefont {Yu}~\bibnamefont {Cai}}, \bibinfo {author} {\bibfnamefont {Yeong-Cherng}\ \bibnamefont {Liang}},\ and\ \bibinfo {author} {\bibfnamefont {Valerio}\ \bibnamefont {Scarani}},\ }\bibfield  {title} {\enquote {\bibinfo {title} {{Geometry of the set of quantum correlations}},}\ }\href {https://doi.org/10.1103/PhysRevA.97.022104} {\bibfield  {journal} {\bibinfo  {journal} {Phys. Rev. A}\ }\textbf {\bibinfo {volume} {97}},\ \bibinfo {pages} {022104} (\bibinfo {year} {2018})}\BibitemShut {NoStop}%
\bibitem [{\citenamefont {Masanes}(2005)}]{masanes2005extremal}%
  \BibitemOpen
  \bibfield  {author} {\bibinfo {author} {\bibfnamefont {Ll.}\ \bibnamefont {Masanes}},\ }\href@noop {} {\enquote {\bibinfo {title} {{Extremal quantum correlations for N parties with two dichotomic observables per site}},}\ } (\bibinfo {year} {2005}),\ \Eprint {https://arxiv.org/abs/quant-ph/0512100} {arXiv:quant-ph/0512100 [quant-ph]} \BibitemShut {NoStop}%
\bibitem [{\citenamefont {Masanes}(2006)}]{masanes2006asymptotic}%
  \BibitemOpen
  \bibfield  {author} {\bibinfo {author} {\bibfnamefont {Llu{\'\i}s}\ \bibnamefont {Masanes}},\ }\bibfield  {title} {\enquote {\bibinfo {title} {{Asymptotic violation of Bell inequalities and distillability}},}\ }\href {https://doi.org/10.1103/PhysRevLett.97.050503} {\bibfield  {journal} {\bibinfo  {journal} {Physical Review Letters}\ }\textbf {\bibinfo {volume} {97}},\ \bibinfo {pages} {050503} (\bibinfo {year} {2006})}\BibitemShut {NoStop}%
\bibitem [{\citenamefont {Pironio}\ \emph {et~al.}(2009)\citenamefont {Pironio}, \citenamefont {Acín}, \citenamefont {Brunner}, \citenamefont {Gisin}, \citenamefont {Massar},\ and\ \citenamefont {Scarani}}]{Pironio_2009}%
  \BibitemOpen
  \bibfield  {author} {\bibinfo {author} {\bibfnamefont {Stefano}\ \bibnamefont {Pironio}}, \bibinfo {author} {\bibfnamefont {Antonio}\ \bibnamefont {Acín}}, \bibinfo {author} {\bibfnamefont {Nicolas}\ \bibnamefont {Brunner}}, \bibinfo {author} {\bibfnamefont {Nicolas}\ \bibnamefont {Gisin}}, \bibinfo {author} {\bibfnamefont {Serge}\ \bibnamefont {Massar}},\ and\ \bibinfo {author} {\bibfnamefont {Valerio}\ \bibnamefont {Scarani}},\ }\bibfield  {title} {\enquote {\bibinfo {title} {Device{\textendash}independent quantum key distribution secure against collective attacks},}\ }\href {https://doi.org/10.1088/1367-2630/11/4/045021} {\bibfield  {journal} {\bibinfo  {journal} {New Journal of Physics}\ }\textbf {\bibinfo {volume} {11}},\ \bibinfo {pages} {045021} (\bibinfo {year} {2009})}\BibitemShut {NoStop}%
\bibitem [{\citenamefont {Irfan}\ \emph {et~al.}(2020)\citenamefont {Irfan}, \citenamefont {Mayer}, \citenamefont {Ortiz},\ and\ \citenamefont {Knill}}]{Irfan_2020}%
  \BibitemOpen
  \bibfield  {author} {\bibinfo {author} {\bibfnamefont {Abu Ashik~Md.}\ \bibnamefont {Irfan}}, \bibinfo {author} {\bibfnamefont {Karl}\ \bibnamefont {Mayer}}, \bibinfo {author} {\bibfnamefont {Gerardo}\ \bibnamefont {Ortiz}},\ and\ \bibinfo {author} {\bibfnamefont {Emanuel}\ \bibnamefont {Knill}},\ }\bibfield  {title} {\enquote {\bibinfo {title} {{Certified quantum measurement of Majorana fermions}},}\ }\href {https://doi.org/10.1103/PhysRevA.101.032106} {\bibfield  {journal} {\bibinfo  {journal} {Phys. Rev. A}\ }\textbf {\bibinfo {volume} {101}},\ \bibinfo {pages} {032106} (\bibinfo {year} {2020})}\BibitemShut {NoStop}%
\bibitem [{\citenamefont {Dykema}\ \emph {et~al.}(2019)\citenamefont {Dykema}, \citenamefont {Paulsen},\ and\ \citenamefont {Prakash}}]{dykema2019non}%
  \BibitemOpen
  \bibfield  {author} {\bibinfo {author} {\bibfnamefont {Ken}\ \bibnamefont {Dykema}}, \bibinfo {author} {\bibfnamefont {Vern~I}\ \bibnamefont {Paulsen}},\ and\ \bibinfo {author} {\bibfnamefont {Jitendra}\ \bibnamefont {Prakash}},\ }\bibfield  {title} {\enquote {\bibinfo {title} {{Non-closure of the set of quantum correlations via graphs}},}\ }\href {https://doi.org/10.1007/s00220-019-03301-1} {\bibfield  {journal} {\bibinfo  {journal} {Communications in Mathematical Physics}\ }\textbf {\bibinfo {volume} {365}},\ \bibinfo {pages} {1125--1142} (\bibinfo {year} {2019})}\BibitemShut {NoStop}%
\bibitem [{\citenamefont {P{\'a}l}\ and\ \citenamefont {V{\'e}rtesi}(2010)}]{pal2010maximal}%
  \BibitemOpen
  \bibfield  {author} {\bibinfo {author} {\bibfnamefont {K{\'a}roly~F}\ \bibnamefont {P{\'a}l}}\ and\ \bibinfo {author} {\bibfnamefont {Tam{\'a}s}\ \bibnamefont {V{\'e}rtesi}},\ }\bibfield  {title} {\enquote {\bibinfo {title} {{Maximal violation of a bipartite three-setting, two-outcome Bell inequality using infinite-dimensional quantum systems}},}\ }\href {https://journals.aps.org/pra/abstract/10.1103/PhysRevA.82.022116} {\bibfield  {journal} {\bibinfo  {journal} {Physical Review A}\ }\textbf {\bibinfo {volume} {82}},\ \bibinfo {pages} {022116} (\bibinfo {year} {2010})}\BibitemShut {NoStop}%
\bibitem [{\citenamefont {Collins}\ and\ \citenamefont {Gisin}(2004)}]{Collins2004}%
  \BibitemOpen
  \bibfield  {author} {\bibinfo {author} {\bibfnamefont {Daniel}\ \bibnamefont {Collins}}\ and\ \bibinfo {author} {\bibfnamefont {Nicolas}\ \bibnamefont {Gisin}},\ }\bibfield  {title} {\enquote {\bibinfo {title} {{A relevant two qubit Bell inequality inequivalent to the CHSH inequality}},}\ }\href {https://doi.org/10.1088/0305-4470/37/5/021} {\bibfield  {journal} {\bibinfo  {journal} {Journal of Physics A: Mathematical and General}\ }\textbf {\bibinfo {volume} {37}},\ \bibinfo {pages} {1775–1787} (\bibinfo {year} {2004})}\BibitemShut {NoStop}%
\bibitem [{\citenamefont {Coladangelo}\ and\ \citenamefont {Stark}(2020)}]{coladangelo2020inherently}%
  \BibitemOpen
  \bibfield  {author} {\bibinfo {author} {\bibfnamefont {Andrea}\ \bibnamefont {Coladangelo}}\ and\ \bibinfo {author} {\bibfnamefont {Jalex}\ \bibnamefont {Stark}},\ }\bibfield  {title} {\enquote {\bibinfo {title} {{An inherently infinite-dimensional quantum correlation}},}\ }\href {https://doi.org/10.1038/s41467-020-17077-9} {\bibfield  {journal} {\bibinfo  {journal} {Nature communications}\ }\textbf {\bibinfo {volume} {11}},\ \bibinfo {pages} {3335} (\bibinfo {year} {2020})}\BibitemShut {NoStop}%
\bibitem [{\citenamefont {Donohue}\ and\ \citenamefont {Wolfe}(2015)}]{nonconvex}%
  \BibitemOpen
  \bibfield  {author} {\bibinfo {author} {\bibfnamefont {John~Matthew}\ \bibnamefont {Donohue}}\ and\ \bibinfo {author} {\bibfnamefont {Elie}\ \bibnamefont {Wolfe}},\ }\bibfield  {title} {\enquote {\bibinfo {title} {{Identifying nonconvexity in the sets of limited-dimension quantum correlations}},}\ }\href {https://doi.org/10.1103/PhysRevA.92.062120} {\bibfield  {journal} {\bibinfo  {journal} {Phys. Rev. A}\ }\textbf {\bibinfo {volume} {92}},\ \bibinfo {pages} {062120} (\bibinfo {year} {2015})}\BibitemShut {NoStop}%
\bibitem [{\citenamefont {Navascu\'es}\ \emph {et~al.}(2007)\citenamefont {Navascu\'es}, \citenamefont {Pironio},\ and\ \citenamefont {Ac\'{\i}n}}]{PhysRevLett.98.010401}%
  \BibitemOpen
  \bibfield  {author} {\bibinfo {author} {\bibfnamefont {Miguel}\ \bibnamefont {Navascu\'es}}, \bibinfo {author} {\bibfnamefont {Stefano}\ \bibnamefont {Pironio}},\ and\ \bibinfo {author} {\bibfnamefont {Antonio}\ \bibnamefont {Ac\'{\i}n}},\ }\bibfield  {title} {\enquote {\bibinfo {title} {{Bounding the Set of Quantum Correlations}},}\ }\href {https://doi.org/10.1103/PhysRevLett.98.010401} {\bibfield  {journal} {\bibinfo  {journal} {Phys. Rev. Lett.}\ }\textbf {\bibinfo {volume} {98}},\ \bibinfo {pages} {010401} (\bibinfo {year} {2007})}\BibitemShut {NoStop}%
\bibitem [{\citenamefont {Collins}\ \emph {et~al.}(2002)\citenamefont {Collins}, \citenamefont {Gisin}, \citenamefont {Linden}, \citenamefont {Massar},\ and\ \citenamefont {Popescu}}]{collins2002bell}%
  \BibitemOpen
  \bibfield  {author} {\bibinfo {author} {\bibfnamefont {Daniel}\ \bibnamefont {Collins}}, \bibinfo {author} {\bibfnamefont {Nicolas}\ \bibnamefont {Gisin}}, \bibinfo {author} {\bibfnamefont {Noah}\ \bibnamefont {Linden}}, \bibinfo {author} {\bibfnamefont {Serge}\ \bibnamefont {Massar}},\ and\ \bibinfo {author} {\bibfnamefont {Sandu}\ \bibnamefont {Popescu}},\ }\bibfield  {title} {\enquote {\bibinfo {title} {{Bell inequalities for arbitrarily high-dimensional systems}},}\ }\href {https://doi.org/10.1103/PhysRevLett.88.040404} {\bibfield  {journal} {\bibinfo  {journal} {Physical review letters}\ }\textbf {\bibinfo {volume} {88}},\ \bibinfo {pages} {040404} (\bibinfo {year} {2002})}\BibitemShut {NoStop}%
\bibitem [{\citenamefont {Grandjean}\ \emph {et~al.}(2012)\citenamefont {Grandjean}, \citenamefont {Liang}, \citenamefont {Bancal}, \citenamefont {Brunner},\ and\ \citenamefont {Gisin}}]{PhysRevA.85.052113}%
  \BibitemOpen
  \bibfield  {author} {\bibinfo {author} {\bibfnamefont {Basile}\ \bibnamefont {Grandjean}}, \bibinfo {author} {\bibfnamefont {Yeong-Cherng}\ \bibnamefont {Liang}}, \bibinfo {author} {\bibfnamefont {Jean-Daniel}\ \bibnamefont {Bancal}}, \bibinfo {author} {\bibfnamefont {Nicolas}\ \bibnamefont {Brunner}},\ and\ \bibinfo {author} {\bibfnamefont {Nicolas}\ \bibnamefont {Gisin}},\ }\bibfield  {title} {\enquote {\bibinfo {title} {{Bell inequalities for three systems and arbitrarily many measurement outcomes}},}\ }\href {https://doi.org/10.1103/PhysRevA.85.052113} {\bibfield  {journal} {\bibinfo  {journal} {Phys. Rev. A}\ }\textbf {\bibinfo {volume} {85}},\ \bibinfo {pages} {052113} (\bibinfo {year} {2012})}\BibitemShut {NoStop}%
\bibitem [{\citenamefont {{Sikora}}\ \emph {et~al.}(2016)\citenamefont {{Sikora}}, \citenamefont {{Varvitsiotis}},\ and\ \citenamefont {{Wei}}}]{Sikora_2016}%
  \BibitemOpen
  \bibfield  {author} {\bibinfo {author} {\bibfnamefont {Jamie}\ \bibnamefont {{Sikora}}}, \bibinfo {author} {\bibfnamefont {Antonios}\ \bibnamefont {{Varvitsiotis}}},\ and\ \bibinfo {author} {\bibfnamefont {Zhaohui}\ \bibnamefont {{Wei}}},\ }\bibfield  {title} {\enquote {\bibinfo {title} {{Minimum Dimension of a Hilbert Space Needed to Generate a Quantum Correlation}},}\ }\href {https://doi.org/10.1103/PhysRevLett.117.060401} {\bibfield  {journal} {\bibinfo  {journal} {\prl}\ }\textbf {\bibinfo {volume} {117}},\ \bibinfo {eid} {060401} (\bibinfo {year} {2016})}\BibitemShut {NoStop}%
\bibitem [{\citenamefont {Pironio}(2005)}]{Lifting}%
  \BibitemOpen
  \bibfield  {author} {\bibinfo {author} {\bibfnamefont {Stefano}\ \bibnamefont {Pironio}},\ }\bibfield  {title} {\enquote {\bibinfo {title} {{Lifting Bell inequalities}},}\ }\bibfield  {journal} {\bibinfo  {journal} {Journal of Mathematical Physics}\ }\textbf {\bibinfo {volume} {46}},\ \href {https://doi.org/10.1063/1.1928727} {10.1063/1.1928727} (\bibinfo {year} {2005})\BibitemShut {NoStop}%
\bibitem [{\citenamefont {Moreno}\ \emph {et~al.}(2020)\citenamefont {Moreno}, \citenamefont {Nery}, \citenamefont {Palhares},\ and\ \citenamefont {Chaves}}]{Chaves}%
  \BibitemOpen
  \bibfield  {author} {\bibinfo {author} {\bibfnamefont {George}\ \bibnamefont {Moreno}}, \bibinfo {author} {\bibfnamefont {Ranieri}\ \bibnamefont {Nery}}, \bibinfo {author} {\bibfnamefont {Alberto}\ \bibnamefont {Palhares}},\ and\ \bibinfo {author} {\bibfnamefont {Rafael}\ \bibnamefont {Chaves}},\ }\bibfield  {title} {\enquote {\bibinfo {title} {{Multistage games and Bell scenarios with communication}},}\ }\href {https://doi.org/10.1103/PhysRevA.102.042412} {\bibfield  {journal} {\bibinfo  {journal} {Phys. Rev. A}\ }\textbf {\bibinfo {volume} {102}},\ \bibinfo {pages} {042412} (\bibinfo {year} {2020})}\BibitemShut {NoStop}%
\bibitem [{\citenamefont {Wolfe}\ \emph {et~al.}(2021{\natexlab{a}})\citenamefont {Wolfe}, \citenamefont {Pozas-Kerstjens}, \citenamefont {Grinberg}, \citenamefont {Rosset}, \citenamefont {Ac{\'\i}n},\ and\ \citenamefont {Navascu{\'e}s}}]{wolfe2021quantum}%
  \BibitemOpen
  \bibfield  {author} {\bibinfo {author} {\bibfnamefont {Elie}\ \bibnamefont {Wolfe}}, \bibinfo {author} {\bibfnamefont {Alejandro}\ \bibnamefont {Pozas-Kerstjens}}, \bibinfo {author} {\bibfnamefont {Matan}\ \bibnamefont {Grinberg}}, \bibinfo {author} {\bibfnamefont {Denis}\ \bibnamefont {Rosset}}, \bibinfo {author} {\bibfnamefont {Antonio}\ \bibnamefont {Ac{\'\i}n}},\ and\ \bibinfo {author} {\bibfnamefont {Miguel}\ \bibnamefont {Navascu{\'e}s}},\ }\bibfield  {title} {\enquote {\bibinfo {title} {{Quantum inflation: A general approach to quantum causal compatibility}},}\ }\href {https://doi.org/10.1103/PhysRevX.11.021043} {\bibfield  {journal} {\bibinfo  {journal} {Physical Review X}\ }\textbf {\bibinfo {volume} {11}},\ \bibinfo {pages} {021043} (\bibinfo {year} {2021}{\natexlab{a}})}\BibitemShut {NoStop}%
\bibitem [{\citenamefont {Van~Himbeeck}\ \emph {et~al.}(2019)\citenamefont {Van~Himbeeck}, \citenamefont {Brask}, \citenamefont {Pironio}, \citenamefont {Ramanathan}, \citenamefont {Sainz},\ and\ \citenamefont {Wolfe}}]{van2019quantum}%
  \BibitemOpen
  \bibfield  {author} {\bibinfo {author} {\bibfnamefont {Thomas}\ \bibnamefont {Van~Himbeeck}}, \bibinfo {author} {\bibfnamefont {Jonatan~Bohr}\ \bibnamefont {Brask}}, \bibinfo {author} {\bibfnamefont {Stefano}\ \bibnamefont {Pironio}}, \bibinfo {author} {\bibfnamefont {Ravishankar}\ \bibnamefont {Ramanathan}}, \bibinfo {author} {\bibfnamefont {Ana~Bel{\'e}n}\ \bibnamefont {Sainz}},\ and\ \bibinfo {author} {\bibfnamefont {Elie}\ \bibnamefont {Wolfe}},\ }\bibfield  {title} {\enquote {\bibinfo {title} {{Quantum violations in the Instrumental scenario and their relations to the Bell scenario}},}\ }\href {https://doi.org/10.22331/q-2019-09-16-186} {\bibfield  {journal} {\bibinfo  {journal} {Quantum}\ }\textbf {\bibinfo {volume} {3}},\ \bibinfo {pages} {186} (\bibinfo {year} {2019})}\BibitemShut {NoStop}%
\bibitem [{\citenamefont {{Sciara}}\ \emph {et~al.}(2017)\citenamefont {{Sciara}}, \citenamefont {{Lo Franco}},\ and\ \citenamefont {{Compagno}}}]{2017NatSR...744675S}%
  \BibitemOpen
  \bibfield  {author} {\bibinfo {author} {\bibfnamefont {Stefania}\ \bibnamefont {{Sciara}}}, \bibinfo {author} {\bibfnamefont {Rosario}\ \bibnamefont {{Lo Franco}}},\ and\ \bibinfo {author} {\bibfnamefont {Giuseppe}\ \bibnamefont {{Compagno}}},\ }\bibfield  {title} {\enquote {\bibinfo {title} {{Universality of Schmidt decomposition and particle identity}},}\ }\href {https://doi.org/10.1038/srep44675} {\bibfield  {journal} {\bibinfo  {journal} {Scientific Reports}\ }\textbf {\bibinfo {volume} {7}},\ \bibinfo {eid} {44675} (\bibinfo {year} {2017})}\BibitemShut {NoStop}%
\bibitem [{\citenamefont {{Almeida}}\ \emph {et~al.}(2010)\citenamefont {{Almeida}}, \citenamefont {{Bancal}}, \citenamefont {{Brunner}}, \citenamefont {{Ac{\'\i}n}}, \citenamefont {{Gisin}},\ and\ \citenamefont {{Pironio}}}]{2010_GYNI}%
  \BibitemOpen
  \bibfield  {author} {\bibinfo {author} {\bibfnamefont {Mafalda~L.}\ \bibnamefont {{Almeida}}}, \bibinfo {author} {\bibfnamefont {Jean-Daniel}\ \bibnamefont {{Bancal}}}, \bibinfo {author} {\bibfnamefont {Nicolas}\ \bibnamefont {{Brunner}}}, \bibinfo {author} {\bibfnamefont {Antonio}\ \bibnamefont {{Ac{\'\i}n}}}, \bibinfo {author} {\bibfnamefont {Nicolas}\ \bibnamefont {{Gisin}}},\ and\ \bibinfo {author} {\bibfnamefont {Stefano}\ \bibnamefont {{Pironio}}},\ }\bibfield  {title} {\enquote {\bibinfo {title} {{Guess Your Neighbor's Input: A Multipartite Nonlocal Game with No Quantum Advantage}},}\ }\href {https://doi.org/10.1103/PhysRevLett.104.230404} {\bibfield  {journal} {\bibinfo  {journal} {\prl}\ }\textbf {\bibinfo {volume} {104}},\ \bibinfo {eid} {230404} (\bibinfo {year} {2010})}\BibitemShut {NoStop}%
\bibitem [{\citenamefont {Navascu{\'e}s}\ \emph {et~al.}(2008)\citenamefont {Navascu{\'e}s}, \citenamefont {Pironio},\ and\ \citenamefont {Ac{\'\i}n}}]{navascues2008convergent}%
  \BibitemOpen
  \bibfield  {author} {\bibinfo {author} {\bibfnamefont {Miguel}\ \bibnamefont {Navascu{\'e}s}}, \bibinfo {author} {\bibfnamefont {Stefano}\ \bibnamefont {Pironio}},\ and\ \bibinfo {author} {\bibfnamefont {Antonio}\ \bibnamefont {Ac{\'\i}n}},\ }\bibfield  {title} {\enquote {\bibinfo {title} {{A convergent hierarchy of semidefinite programs characterizing the set of quantum correlations}},}\ }\href {https://iopscience.iop.org/article/10.1088/1367-2630/10/7/073013} {\bibfield  {journal} {\bibinfo  {journal} {New Journal of Physics}\ }\textbf {\bibinfo {volume} {10}},\ \bibinfo {pages} {073013} (\bibinfo {year} {2008})}\BibitemShut {NoStop}%
\bibitem [{\citenamefont {Poderini}\ \emph {et~al.}(2020)\citenamefont {Poderini}, \citenamefont {Chaves}, \citenamefont {Agresti}, \citenamefont {Carvacho},\ and\ \citenamefont {Sciarrino}}]{poderini2020exclusivity}%
  \BibitemOpen
  \bibfield  {author} {\bibinfo {author} {\bibfnamefont {Davide}\ \bibnamefont {Poderini}}, \bibinfo {author} {\bibfnamefont {Rafael}\ \bibnamefont {Chaves}}, \bibinfo {author} {\bibfnamefont {Iris}\ \bibnamefont {Agresti}}, \bibinfo {author} {\bibfnamefont {Gonzalo}\ \bibnamefont {Carvacho}},\ and\ \bibinfo {author} {\bibfnamefont {Fabio}\ \bibnamefont {Sciarrino}},\ }\bibfield  {title} {\enquote {\bibinfo {title} {{Exclusivity graph approach to instrumental inequalities}},}\ }in\ \href {https://arxiv.org/abs/1909.09120} {\emph {\bibinfo {booktitle} {Uncertainty in Artificial Intelligence}}}\ (\bibinfo {organization} {PMLR},\ \bibinfo {year} {2020})\ pp.\ \bibinfo {pages} {1274--1283}\BibitemShut {NoStop}%
\bibitem [{\citenamefont {Leonard}\ and\ \citenamefont {Lewis}(2015)}]{leonard2015geometry}%
  \BibitemOpen
  \bibfield  {author} {\bibinfo {author} {\bibfnamefont {Isaac~E}\ \bibnamefont {Leonard}}\ and\ \bibinfo {author} {\bibfnamefont {James~Edward}\ \bibnamefont {Lewis}},\ }\href {https://www.wiley.com/en-ca/Geometry+of+Convex+Sets-p-9781119022664} {\emph {\bibinfo {title} {{Geometry of convex sets}}}}\ (\bibinfo  {publisher} {John Wiley \& Sons},\ \bibinfo {year} {2015})\BibitemShut {NoStop}%
\bibitem [{\citenamefont {Pironio}(2014)}]{pironio2014all}%
  \BibitemOpen
  \bibfield  {author} {\bibinfo {author} {\bibfnamefont {Stefano}\ \bibnamefont {Pironio}},\ }\bibfield  {title} {\enquote {\bibinfo {title} {{All Clauser--Horne--Shimony--Holt polytopes}},}\ }\href {https://iopscience.iop.org/article/10.1088/1751-8113/47/42/424020} {\bibfield  {journal} {\bibinfo  {journal} {Journal of Physics A: Mathematical and Theoretical}\ }\textbf {\bibinfo {volume} {47}},\ \bibinfo {pages} {424020} (\bibinfo {year} {2014})}\BibitemShut {NoStop}%
\bibitem [{\citenamefont {Karczewski}\ \emph {et~al.}(2022)\citenamefont {Karczewski}, \citenamefont {Scala}, \citenamefont {Mandarino}, \citenamefont {Sainz},\ and\ \citenamefont {{\.Z}ukowski}}]{karczewski2022avenues}%
  \BibitemOpen
  \bibfield  {author} {\bibinfo {author} {\bibfnamefont {Marcin}\ \bibnamefont {Karczewski}}, \bibinfo {author} {\bibfnamefont {Giovanni}\ \bibnamefont {Scala}}, \bibinfo {author} {\bibfnamefont {Antonio}\ \bibnamefont {Mandarino}}, \bibinfo {author} {\bibfnamefont {Ana~Bel{\'e}n}\ \bibnamefont {Sainz}},\ and\ \bibinfo {author} {\bibfnamefont {Marek}\ \bibnamefont {{\.Z}ukowski}},\ }\bibfield  {title} {\enquote {\bibinfo {title} {{Avenues to generalising Bell inequalities}},}\ }\href {https://iopscience.iop.org/article/10.1088/1751-8121/ac8a28} {\bibfield  {journal} {\bibinfo  {journal} {Journal of Physics A: Mathematical and Theoretical}\ }\textbf {\bibinfo {volume} {55}},\ \bibinfo {pages} {384011} (\bibinfo {year} {2022})}\BibitemShut {NoStop}%
\bibitem [{\citenamefont {Bernards}\ and\ \citenamefont {G{\"u}hne}(2021)}]{bernards2021finding}%
  \BibitemOpen
  \bibfield  {author} {\bibinfo {author} {\bibfnamefont {Fabian}\ \bibnamefont {Bernards}}\ and\ \bibinfo {author} {\bibfnamefont {Otfried}\ \bibnamefont {G{\"u}hne}},\ }\bibfield  {title} {\enquote {\bibinfo {title} {{Finding optimal Bell inequalities using the cone-projection technique}},}\ }\href {https://doi.org/10.1103/PhysRevA.104.012206} {\bibfield  {journal} {\bibinfo  {journal} {Physical Review A}\ }\textbf {\bibinfo {volume} {104}},\ \bibinfo {pages} {012206} (\bibinfo {year} {2021})}\BibitemShut {NoStop}%
\bibitem [{\citenamefont {P{\'a}l}\ and\ \citenamefont {V{\'e}rtesi}(2017)}]{pal2017family}%
  \BibitemOpen
  \bibfield  {author} {\bibinfo {author} {\bibfnamefont {K{\'a}roly~F}\ \bibnamefont {P{\'a}l}}\ and\ \bibinfo {author} {\bibfnamefont {Tam{\'a}s}\ \bibnamefont {V{\'e}rtesi}},\ }\bibfield  {title} {\enquote {\bibinfo {title} {{Family of Bell inequalities violated by higher-dimensional bound entangled states}},}\ }\href {https://journals.aps.org/pra/abstract/10.1103/PhysRevA.96.022123} {\bibfield  {journal} {\bibinfo  {journal} {Physical Review A}\ }\textbf {\bibinfo {volume} {96}},\ \bibinfo {pages} {022123} (\bibinfo {year} {2017})}\BibitemShut {NoStop}%
\bibitem [{\citenamefont {P{\'e}rez-Garc{\'\i}a}\ \emph {et~al.}(2008)\citenamefont {P{\'e}rez-Garc{\'\i}a}, \citenamefont {Wolf}, \citenamefont {Palazuelos}, \citenamefont {Villanueva},\ and\ \citenamefont {Junge}}]{perez2008unbounded}%
  \BibitemOpen
  \bibfield  {author} {\bibinfo {author} {\bibfnamefont {David}\ \bibnamefont {P{\'e}rez-Garc{\'\i}a}}, \bibinfo {author} {\bibfnamefont {Michael~M}\ \bibnamefont {Wolf}}, \bibinfo {author} {\bibfnamefont {Carlos}\ \bibnamefont {Palazuelos}}, \bibinfo {author} {\bibfnamefont {Ignacio}\ \bibnamefont {Villanueva}},\ and\ \bibinfo {author} {\bibfnamefont {Marius}\ \bibnamefont {Junge}},\ }\bibfield  {title} {\enquote {\bibinfo {title} {{Unbounded violation of tripartite Bell inequalities}},}\ }\href {https://doi.org/10.48550/arXiv.quant-ph/0702189} {\bibfield  {journal} {\bibinfo  {journal} {Communications in Mathematical Physics}\ }\textbf {\bibinfo {volume} {279}},\ \bibinfo {pages} {455--486} (\bibinfo {year} {2008})}\BibitemShut {NoStop}%
\bibitem [{\citenamefont {Zohren}\ \emph {et~al.}(2010)\citenamefont {Zohren}, \citenamefont {Reska}, \citenamefont {Gill},\ and\ \citenamefont {Westra}}]{zohren2010tight}%
  \BibitemOpen
  \bibfield  {author} {\bibinfo {author} {\bibfnamefont {Stefan}\ \bibnamefont {Zohren}}, \bibinfo {author} {\bibfnamefont {Paul}\ \bibnamefont {Reska}}, \bibinfo {author} {\bibfnamefont {Richard~D}\ \bibnamefont {Gill}},\ and\ \bibinfo {author} {\bibfnamefont {Willem}\ \bibnamefont {Westra}},\ }\bibfield  {title} {\enquote {\bibinfo {title} {{A tight Tsirelson inequality for infinitely many outcomes}},}\ }\href {https://iopscience.iop.org/article/10.1209/0295-5075/90/10002/meta} {\bibfield  {journal} {\bibinfo  {journal} {Europhysics Letters}\ }\textbf {\bibinfo {volume} {90}},\ \bibinfo {pages} {10002} (\bibinfo {year} {2010})}\BibitemShut {NoStop}%
\bibitem [{\citenamefont {Evans}(2018)}]{Evans_2018}%
  \BibitemOpen
  \bibfield  {author} {\bibinfo {author} {\bibfnamefont {Robin~J.}\ \bibnamefont {Evans}},\ }\bibfield  {title} {\enquote {\bibinfo {title} {{Margins of discrete Bayesian networks}},}\ }\href {https://doi.org/10.1214/17-aos1631} {\bibfield  {journal} {\bibinfo  {journal} {The Annals of Statistics}\ }\textbf {\bibinfo {volume} {46}} (\bibinfo {year} {2018})}\BibitemShut {NoStop}%
\bibitem [{\citenamefont {Rosset}\ \emph {et~al.}(2018)\citenamefont {Rosset}, \citenamefont {Gisin},\ and\ \citenamefont {Wolfe}}]{Rosset_2018}%
  \BibitemOpen
  \bibfield  {author} {\bibinfo {author} {\bibfnamefont {Denis}\ \bibnamefont {Rosset}}, \bibinfo {author} {\bibfnamefont {Nicolas}\ \bibnamefont {Gisin}},\ and\ \bibinfo {author} {\bibfnamefont {Elie}\ \bibnamefont {Wolfe}},\ }\bibfield  {title} {\enquote {\bibinfo {title} {{Universal bound on the cardinality of local hidden variables in networks}},}\ }\href {https://doi.org/10.26421/qic18.11-12} {\bibfield  {journal} {\bibinfo  {journal} {Quantum Information and Computation}\ }\textbf {\bibinfo {volume} {18}} (\bibinfo {year} {2018})}\BibitemShut {NoStop}%
\bibitem [{\citenamefont {Zhang}\ \emph {et~al.}(2022)\citenamefont {Zhang}, \citenamefont {Tian},\ and\ \citenamefont {Bareinboim}}]{Bareinboim_2021}%
  \BibitemOpen
  \bibfield  {author} {\bibinfo {author} {\bibfnamefont {Junzhe}\ \bibnamefont {Zhang}}, \bibinfo {author} {\bibfnamefont {Jin}\ \bibnamefont {Tian}},\ and\ \bibinfo {author} {\bibfnamefont {Elias}\ \bibnamefont {Bareinboim}},\ }\bibfield  {title} {\enquote {\bibinfo {title} {{Partial Counterfactual Identification from Observational and Experimental Data}},}\ }in\ \href {https://proceedings.mlr.press/v162/zhang22ab.html} {\emph {\bibinfo {booktitle} {Proceedings of the 39th International Conference on Machine Learning}}},\ \bibinfo {series} {Proceedings of Machine Learning Research}, Vol.\ \bibinfo {volume} {162}\ (\bibinfo  {publisher} {PMLR},\ \bibinfo {year} {2022})\ pp.\ \bibinfo {pages} {26548--26558}\BibitemShut {NoStop}%
\bibitem [{\citenamefont {Wolfe}\ \emph {et~al.}(2021{\natexlab{b}})\citenamefont {Wolfe}, \citenamefont {Pozas-Kerstjens}, \citenamefont {Grinberg}, \citenamefont {Rosset}, \citenamefont {Ac\'{\i}n},\ and\ \citenamefont {Navascu\'es}}]{PhysRevX.11.021043}%
  \BibitemOpen
  \bibfield  {author} {\bibinfo {author} {\bibfnamefont {Elie}\ \bibnamefont {Wolfe}}, \bibinfo {author} {\bibfnamefont {Alejandro}\ \bibnamefont {Pozas-Kerstjens}}, \bibinfo {author} {\bibfnamefont {Matan}\ \bibnamefont {Grinberg}}, \bibinfo {author} {\bibfnamefont {Denis}\ \bibnamefont {Rosset}}, \bibinfo {author} {\bibfnamefont {Antonio}\ \bibnamefont {Ac\'{\i}n}},\ and\ \bibinfo {author} {\bibfnamefont {Miguel}\ \bibnamefont {Navascu\'es}},\ }\bibfield  {title} {\enquote {\bibinfo {title} {{Quantum Inflation: A General Approach to Quantum Causal Compatibility}},}\ }\href {https://doi.org/10.1103/PhysRevX.11.021043} {\bibfield  {journal} {\bibinfo  {journal} {Phys. Rev. X}\ }\textbf {\bibinfo {volume} {11}},\ \bibinfo {pages} {021043} (\bibinfo {year} {2021}{\natexlab{b}})}\BibitemShut {NoStop}%
\bibitem [{\citenamefont {P\'al}\ and\ \citenamefont {V\'ertesi}(2008)}]{PhysRevA.77.042105}%
  \BibitemOpen
  \bibfield  {author} {\bibinfo {author} {\bibfnamefont {K\'aroly~F.}\ \bibnamefont {P\'al}}\ and\ \bibinfo {author} {\bibfnamefont {Tam\'as}\ \bibnamefont {V\'ertesi}},\ }\bibfield  {title} {\enquote {\bibinfo {title} {{Efficiency of higher-dimensional Hilbert spaces for the violation of Bell inequalities}},}\ }\href {https://doi.org/10.1103/PhysRevA.77.042105} {\bibfield  {journal} {\bibinfo  {journal} {Phys. Rev. A}\ }\textbf {\bibinfo {volume} {77}},\ \bibinfo {pages} {042105} (\bibinfo {year} {2008})}\BibitemShut {NoStop}%
\bibitem [{\citenamefont {Brask}\ and\ \citenamefont {Chaves}(2017)}]{Brask_2017}%
  \BibitemOpen
  \bibfield  {author} {\bibinfo {author} {\bibfnamefont {J~B}\ \bibnamefont {Brask}}\ and\ \bibinfo {author} {\bibfnamefont {R}~\bibnamefont {Chaves}},\ }\bibfield  {title} {\enquote {\bibinfo {title} {{Bell scenarios with communication}},}\ }\href {https://doi.org/10.1088/1751-8121/aa5840} {\bibfield  {journal} {\bibinfo  {journal} {Journal of Physics A: Mathematical and Theoretical}\ }\textbf {\bibinfo {volume} {50}},\ \bibinfo {pages} {094001} (\bibinfo {year} {2017})}\BibitemShut {NoStop}%
\end{thebibliography}%

\end{document}